\documentclass[12pt, a4paper]{article}
\usepackage[margin=0.8in]{geometry}
\usepackage{amsthm, amsmath, amssymb, bm, bbm, ascmac, listings, algorithm, algpseudocode, tikz, physics}
\usepackage[backend=biber, style=apa, citestyle=authoryear]{biblatex}
\renewbibmacro*{cite}{%
  \printnames{labelname}%
  \addspace
  \printtext[parens]{\printfield{year}}}
\usetikzlibrary{fit, arrows.meta, trees}
\newtheorem{prop}{Proposition}
\newtheorem{lem}{Lemma}
\newtheorem{thm}{Theorem}

\newtheorem{claim}{Claim}

\title{A Note on `The Limits of Price Discrimination' by Bergemann, Brooks, and Morris\footnote{
I am deeply grateful to my advisor, Michihiro Kandori, and to Stephen Morris for their helpful comments.
}
}
\author{
Keita Kuwahara\thanks{Graduate School of Economics, The University of Tokyo. 
Email: \texttt{kuke0303@g.ecc.u-tokyo.ac.jp}}
}

\date{\today}

\begin{document}

\maketitle

\begin{abstract}
This note revisits the analysis of third-degree price discrimination developed by \cite{bergemann2015limits}, which characterizes the set of consumer-producer surplus pairs that can be achieved through market segmentation. This was proved by means of market segmentation with random prices, but it was claimed that any segmentation with possibly random pricing has a corresponding direct segmentation, where a deterministic price is charged in each market segment.  However, the latter claim is not correct under the definition of market segmentation given in the paper, and we provide counterexamples. We then propose an alternative definition to resolve this issue and examine the implications of the difference between the two definitions in terms of the main result of their paper.
\end{abstract}

\section{Introduction}
\cite{bergemann2015limits} analyze third-degree price discrimination, in which a monopolist segments the aggregate market by acquiring additional information about consumer characteristics and sets different prices for different segments. In particular, they characterize the set of consumer–producer surplus pairs that can be achieved through such segmentation.
They establish this main result using segmentations where the monopolist is allowed to employ stochastic pricing strategies. Nevertheless, Proposition 2 of \cite{bergemann2015limits} states that for any segmentation, even if the monopolist adopts a stochastic pricing strategy, there exists a corresponding direct segmentation. In the direct segmentation, a deterministic price is assigned to each market segment, and the same outcome is achieved.
\par

In this note, however, we present counterexamples to Proposition 2 of \cite{bergemann2015limits} (Section \ref{sec: 観察}). Specifically, we provide two counterexamples: one where the monopolist adopts a stochastic pricing strategy, and one where the strategy is deterministic.  
These counterexamples imply that Proposition 2 of \cite{bergemann2015limits} does not hold in general, and thus one cannot restrict attention to direct segmentations without loss of generality.
To resolve this issue, we proceed to discuss how one might modify the definition of segmentation (Section \ref{sec: revise definition}).
\par

Although Proposition 2 of \cite{bergemann2015limits} is not true under the original definition of direct segmentation, the original definition might ``practically work'' in applications. In particular, we focus on the main result of the paper and show that almost all consumer-producer surplus pairs achievable through general segmentation are also achievable through direct segmentation under the original definition (Section \ref{sec: surplus triangle}). Moreover, it is shown that all such pairs can be exactly achieved, except in the very special case where the demand in the aggregate market has unit elasticity. (i.e., the monopolist is indifferent between charging any one of the consumers' willingness-to-pay levels).

\section{Model}\label{sec: model}
We adopt the model framework introduced by \cite{bergemann2015limits}, and use the same notation throughout. For clarity, we briefly summarize the setting below.
A monopolist sells a good to a continuum of consumers, each demanding a single unit. Normalize the total mass of consumers to one and the constant marginal cost to zero. Consumers may have one of $K \geq 2$ valuations, given by
\[
V \triangleq \{v_1, \ldots, v_k, \ldots, v_K\} \subset \mathbb{R}_+, \quad \text{with } 0 < v_1<\cdots<v_k < \cdots < v_K.
\]
A \emph{market} $x$ is a probability distribution over these $K$ valuations. The set of all markets is:
\[
X \triangleq \Delta(V) = \left\{ x \in \mathbb{R}_+^V \,\middle|\, \sum_{k=1}^K x(v_k) = 1 \right\}.
\]
We denote $x_k \triangleq x(v_k)$, the share of consumers with valuation $v_k$.
Price $v_k$ is \emph{optimal} for market $x$ if it maximizes expected revenue:
\[
v_k \sum_{j=k}^K x_j \geq v_i \sum_{j=i}^K x_j, \quad \text{for all } i = 1, \ldots, K.
\]
Let $X_k\subset X$ denote the set of markets where $v_k$ is optimal.\par

We fix an \emph{aggregate market} $x^* \in X$ and assume $x_k^*>0$ for all $k = 1, \ldots, K$.
A \emph{segmentation} divides the aggregate market into submarkets. Formally, a segmentation $\sigma$ is a probability distribution over $X$. The set of feasible segmentations is
\[
\Sigma \triangleq \left\{ \sigma \in \Delta(X) \,\middle|\, \sum_{x \in \operatorname{supp} \sigma} \sigma(x) \cdot x = x^*,\, |\operatorname{supp} \sigma| < \infty \right\}.
\]
A \emph{pricing rule} for a segmentation $\sigma$ assigns a price distribution to each market in the support of $\sigma$:
\[
\phi : \operatorname{supp} \sigma \to \Delta(V).
\]
Let $\phi_k(x)$ denote the probability of charging price $v_k$ in market $x$. A pricing rule is \emph{optimal} if for each $x$, $v_k \in \operatorname{supp} \phi(x)$ implies $x \in X_k$.
\par

We define a \emph{direct segmentation} as a segmentation that has support on at most $K$ markets, indexed by $k \in A\subset \{1, \ldots, K\}$, such that $x^k \in X_k$.
The \emph{direct pricing rule} assigns probability one to charging price $v_k$ in market $x^k$, that is,
\[
\phi_k(x^k) = 1.
\]
By definition, the direct pricing rule is optimal for direct segmentations, and whenever we refer to a direct segmentation in the subsequent discussion, it is assumed that the monopolist will use direct pricing.
\par

Given a segmentation $\sigma$ and pricing rule $\phi$, consumer surplus is
\[
\sum_{x \in \operatorname{supp} \sigma} \sigma(x) 
\sum_{k=1}^K \phi_k(x) 
\sum_{j=k}^K (v_j - v_k)x_j,
\]
producer surplus is
\[
\sum_{x \in \operatorname{supp} \sigma} \sigma(x) 
\sum_{k=1}^K \phi_k(x) v_k 
\sum_{j=k}^K x_j,
\]
and the total surplus is
\[
\sum_{x \in \operatorname{supp} \sigma} \sigma(x) 
\sum_{k=1}^K \phi_k(x) 
\sum_{j=k}^K v_j x_j.
\]
In particular, the maximum feasible surplus is given by
\[
w^* \triangleq \sum_{j=1}^K v_j x_j^*,
\]
while the producer surplus under uniform pricing is
\[
\pi^* \triangleq \max_{k \in \{1, \ldots, K\}} v_k \sum_{j=k}^K x_j^*.
\]\par

Finally, we define a class of characteristic markets.
Let $\mathcal{V}$ denote the set of all subsets of $V$, and let $\mathcal{V}_k \subset \mathcal{V}$ be the collection of subsets containing $v_k$.  
For each $S \in \mathcal{V}_k$, define a market $x^S$ as the solution to the following conditions:
\begin{equation*}
v_i \sum_{j=i}^K x_j^S = v_k \sum_{j=k}^K x_j^S \quad \text{for all } v_i \in S,
\end{equation*}
and
\begin{equation*}
\sum_{\{j \mid v_j \in S\}} x_j^S = 1.
\end{equation*}
Note that $x^S$ is uniquely determined.
The following result is known from \cite{bergemann2015limits}.
\begin{lem}[Lemma 1 of \cite{bergemann2015limits}]\label{lem: convex hull}
$X_k$ is equal to the convex hull of $\{x^S\}_{S \in \mathcal{V}_k}$
\end{lem}

\section{Proposition 2 of \cite{bergemann2015limits}}\label{sec: 観察}
The following claim is a part of the statement of Proposition 2 of \cite{bergemann2015limits}.
\begin{claim}[Part of Proposition 2 of \cite{bergemann2015limits}]\label{prop: incorrect claim}
For any segmentation and optimal pricing rule $(\sigma, \phi)$, there exists a direct segmentation $\sigma'$ (and associated direct pricing rule $\phi'$) 
that achieves the same joint distribution over valuations and prices. 
As such, it achieves the same producer surplus, consumer surplus, and total surplus.
\end{claim}
However, this claim is incorrect. The following simple counterexample shows that it does not hold in general.\\[1em]
\textbf{Example 1:}
Let $V = \{1, 2\}$ and $x^* = \left( \frac{1}{2}, \frac{1}{2} \right)$. Define $\sigma$ and $\phi$ as follows:
\[
\sigma\left(x^*\right) = 1, \quad
\phi_1\left(x^*\right) = \frac{1}{2}, \quad
\phi_2\left(x^*\right) = \frac{1}{2}.
\]
Note that $\phi$ is an optimal pricing rule.
We show that there is no direct segmentation that achieves the same joint distribution over valuations and prices as $(\sigma, \phi)$.
To derive a contradiction, suppose there exists a direct segmentation $\sigma'$ and an associated direct pricing rule $\phi'$ that achieve the same joint distribution.\par
Since the price under $(\sigma, \phi)$ can take the values $1$ or $2$, and the marginal distribution of prices under $(\sigma', \phi')$ must coincide with that under $(\sigma, \phi)$, it follows that $\sigma'(x^1) > 0$, $\sigma'(x^2) > 0$ for some $x^1$ and $x^2$.
Note that $\phi'_1(x^1) = 1$ and $\phi'_2(x^2) = 1$.  
Moreover, for $k = 1, 2$, the distribution over valuations conditional on price $v_k$ under $(\sigma, \phi)$ must coincide with that under $(\sigma', \phi')$, that is, $x^k$.
Now, under $(\sigma, \phi)$, the distribution over valuations conditional on price $1$ is
\[
\frac{1}{\sum_{x \in \operatorname{supp} \sigma} \sigma(x) \phi_1(x)} \sum_{x \in \operatorname{supp} \sigma} \sigma(x) \phi_1(x) \cdot x = x^*.
\]
Similarly, the distribution over valuations conditional on price $2$ is
\[
\frac{1}{\sum_{x \in \operatorname{supp} \sigma} \sigma(x) \phi_2(x)} \sum_{x \in \operatorname{supp} \sigma} \sigma(x) \phi_2(x) \cdot x = x^*.
\]
Hence, we must have $x^1 = x^2 = x^*$. However, this leads to a contradiction:
\[
\phi'(x^1) \neq \phi'(x^2) = \phi'(x^1).
\]\\\par
One might suspect that, in Example 1, the failure to construct a corresponding direct segmentation was due to the fact that the monopolist's pricing strategy is stochastic.
However, this is not the case.  
The following example demonstrates that even when the pricing strategy is deterministic, there may still be no corresponding direct segmentation that yields the same joint distribution over valuations and prices.
\\[1em]
\textbf{Example 2:} Let $V = \{1, 2, 3\}$ and $x^* = \left( \frac{1}{2}, \frac{1}{3}, \frac{1}{6} \right)$. Define $\sigma$ and $\phi$ as follows:
\begin{alignat*}{2}
\sigma\left(\left(\frac{1}{2}, \frac{1}{6}, \frac{1}{3}\right)\right) &= \frac{1}{4}, \quad
\sigma\left(\left(\frac{1}{2}, \frac{1}{2}, 0\right)\right) &= \frac{1}{4}, \quad
\sigma\left(\left(\frac{1}{2}, \frac{1}{3}, \frac{1}{6}\right)\right) &= \frac{1}{2}, \\
\phi_1\left(\left(\frac{1}{2}, \frac{1}{6}, \frac{1}{3}\right)\right) &= 1, \quad
\phi_1\left(\left(\frac{1}{2}, \frac{1}{2}, 0\right)\right) &= 1, \quad
\phi_2\left(\left(\frac{1}{2}, \frac{1}{3}, \frac{1}{6}\right)\right) &= 1.
\end{alignat*}
Note that $\phi$ is an optimal pricing rule.
Similar to Example 1, we show that there is no direct segmentation that achieves the same joint distribution over valuations and prices as $(\sigma, \phi)$.
To derive a contradiction, suppose there exists a direct segmentation $\sigma'$ and an associated direct pricing rule $\phi'$ that achieve the same joint distribution.\par
Since the price under $(\sigma, \phi)$ can only take the values $1$ or $2$, and the marginal distribution of prices under $(\sigma', \phi')$ must coincide with that under $(\sigma, \phi)$, it follows that $\sigma'(x^1) > 0$, $\sigma'(x^2) > 0$, and $\sigma'(x^1) + \sigma'(x^2) = 1$ for some $x^1$ and $x^2$.
Note that $\phi'_1(x^1) = 1$ and $\phi'_2(x^2) = 1$.  
Moreover, for $k = 1, 2$, the distribution over valuations conditional on price $v_k$ under $(\sigma, \phi)$ must coincide with that under $(\sigma', \phi')$, that is, $x^k$.
Now, under $(\sigma, \phi)$, the distribution over valuations conditional on price $1$ is
\[
\frac{1}{\sum_{x \in \operatorname{supp} \sigma} \sigma(x) \phi_1(x)} \sum_{x \in \operatorname{supp} \sigma} \sigma(x) \phi_1(x) \cdot x = 2 \left( \frac{1}{4} \left( \frac{1}{2}, \frac{1}{6}, \frac{1}{3} \right) + \frac{1}{4} \left( \frac{1}{2}, \frac{1}{2}, 0 \right) \right) = \left( \frac{1}{2}, \frac{1}{3}, \frac{1}{6} \right).
\]
Similarly, the distribution over valuations conditional on price $2$ is
\[
\frac{1}{\sum_{x \in \operatorname{supp} \sigma} \sigma(x) \phi_2(x)} \sum_{x \in \operatorname{supp} \sigma} \sigma(x) \phi_2(x) \cdot x = 2 \left( \frac{1}{2} \left( \frac{1}{2}, \frac{1}{3}, \frac{1}{6} \right) \right) = \left( \frac{1}{2}, \frac{1}{3}, \frac{1}{6} \right).
\]
Hence, we must have $x^1 = x^2 = \left( \frac{1}{2}, \frac{1}{3}, \frac{1}{6} \right)$. However, this leads to a contradiction:
\[
\phi'(x^1) \neq \phi'(x^2) = \phi'(x^1).
\]\\\par

Building on Examples 1 and 2, we now consider the general case for any $(\sigma, \phi)$.
In the proof of Proposition 2 of \cite{bergemann2015limits}, the direct segmentation $\sigma'$ is constructed as follows:
\[
\sigma'(x^k) \triangleq \sum_{x \in \operatorname{supp} \sigma} \sigma(x) \phi_k(x),
\]
and hence,
\[
x^k \triangleq \frac{1}{\sigma'(x^k)} \sum_{x \in \operatorname{supp} \sigma} \sigma(x) \phi_k(x) \cdot x.
\]
As demonstrated in Examples 1 and 2, if a corresponding direct segmentation exists, then $x^k$ supported by it must necessarily take this specific form.
Therefore, if the resulting $\{x^k\}$ constructed in this way are not all distinct, no corresponding direct segmentation exists.
Conversely, if the $\{x^k\}$ are all distinct, then the direct segmentation $\sigma'$ constructed in this way yields the same joint distribution over valuations and prices.\par

\section{Proposal for Revised Definitions}\label{sec: revise definition}
One way to make Claim \ref{prop: incorrect claim} true is to appropriately revise the relevant definitions. 
In this section, we consider modifying the definition of segmentation. As discussed in Section \ref{sec: 観察}, the root cause of the issue lies in the fact that the monopolist cannot distinguish between two market segments that share the same distribution of valuations. To address this, we redefine a segmentation $\sigma$ as a distribution over $X \times V$ satisfying the following conditions:
\[
\sum_{(x, v_k) \in \operatorname{supp} \sigma} \sigma((x, v_k)) \cdot x = x^* \quad \text{and} \quad |\operatorname{supp} \sigma| < \infty.
\]
In addition, we redefine $\phi : \operatorname{supp} \sigma \to \Delta(V)$ to satisfy the following condition: for all $(x, v_k) \in \operatorname{supp} \sigma$, we have $$\phi_k((x, v_k)) = 1.$$
Note that, given $\sigma$, $\phi$ is uniquely determined by it.
For any $(\sigma', \phi')$ under the original definition, one can construct $\sigma \in \Delta(X \times V)$ under the revised definition so as to achieve the same joint distribution over valuations and prices by setting
\[
\sigma((x, v_k)) \triangleq \sigma'(x) \phi'_k(x).
\]
Therefore, segmentation under this definition is as flexible as that under the original definition.
We also note that this modification does not affect the main result of \cite{bergemann2015limits} (Theorem 1 in the paper).\par

Since we have modified the definition of segmentation, we also redefine direct segmentation for consistency.  
A direct segmentation $\sigma$ is a segmentation supported on at most $K$ markets, indexed by $k \in A \subset \{1, \ldots, K\}$, such that $(x^k, v_k) \in \operatorname{supp} \sigma$.
Under this definition, Claim \ref{prop: incorrect claim} becomes true. For any segmentation $\sigma$, a corresponding direct segmentation $\sigma'$ can be constructed as follows:
\[
\sigma'((x^k, v_k)) \triangleq \sum_{(x, v_\ell) \in \operatorname{supp} \sigma\,:\, v_\ell = v_k} \sigma((x, v_\ell)),
\]
and hence,
\[
x^k \triangleq \frac{1}{\sigma'((x^k, v_k))} \sum_{(x, v_\ell) \in \operatorname{supp} \sigma\,:\, v_\ell = v_k} \sigma((x, v_\ell)) \cdot x.
\]
In both Example 1 and 2, we can construct the direct segmentation $\sigma'$ by setting $\sigma'((x^*, 1)) = \sigma'((x^*, 2)) = \frac{1}{2}$.
\par

The above discussion has focused on modifications to the definitions. In contrast, Section \ref{sec: surplus triangle} presents results concerning what happens if we proceed with the analysis without altering the original definitions.

\section{Surplus Triangle}\label{sec: surplus triangle}
The main result of \cite{bergemann2015limits}, stated as Theorem 1, is as follows.
\begin{thm}[Theorem 1 of \cite{bergemann2015limits}]\label{thm: Surplus Triangle in bergemann2015limits}
There exists a segmentation and optimal pricing rule with consumer surplus $u$ 
and producer surplus $\pi$ if and only if $u \geq 0$, $\pi \geq \pi^*$, 
and $u + \pi \leq w^*$.
\end{thm}
Let $T \triangleq \{(u, \pi) : u \geq 0,\ \pi \geq \pi^*,\ u + \pi \leq w^*\}$.
If Claim \ref{prop: incorrect claim} were correct, then by Theorem \ref{thm: Surplus Triangle in bergemann2015limits}, the set of consumer–producer surplus pairs achievable through direct segmentation would coincide with $T$.
However, under the definitions provided in Section \ref{sec: model}, we have found counterexamples to Claim \ref{prop: incorrect claim}. This naturally raises the following question:  
Does the set of achievable $(u, \pi)$ pairs under direct segmentation still coincide with $T$?
The following theorem addresses this question.
\begin{thm}\label{thm: direct achievable condition}
The set of consumer–producer surplus pairs achievable through direct segmentation coincides with $T$ if and only if $x^* \neq x^V$.
\end{thm}
When $x^* = x^V$, the set of consumer–producer surplus pairs achievable through direct segmentation is characterized by the following proposition.
\begin{prop}\label{prop: direct achievable when $x^* = x^V$}
Let $x^* = x^V$.  
There exists a direct segmentation with consumer surplus $u$ and producer surplus $\pi$ if and only if either:
\begin{enumerate}
    \item $u \geq 0$, $\pi > \pi^*$, and $u + \pi \leq w^*$, or
    \item $\pi = \pi^*$ and
    \[
    u = \sum_{j = k}^K (v_j - v_k) x^*_j \quad \text{for some } k.
    \]
\end{enumerate}
\end{prop}
The proofs of Theorem \ref{thm: direct achievable condition} and Proposition \ref{prop: direct achievable when $x^* = x^V$} are provided in the appendix.

\appendix
\section{Proofs}
\subsection{Proof of Proposition \ref{prop: direct achievable when $x^* = x^V$}}
To prove Proposition \ref{prop: direct achievable when $x^* = x^V$}, we establish Lemmas \ref{lem: pi > pi star} and \ref{lem: K points when x^*=x^V}.
\begin{lem}\label{lem: pi > pi star}
For any $x^*$, if $u \geq 0$, $\pi > \pi^*$, and $u + \pi \leq w^*$, then there exists a direct segmentation with consumer surplus $u$ and producer surplus $\pi$.
\end{lem}
\begin{proof}
We consider two segmentations, both of which are also used in \cite{bergemann2015limits}.
The first is defined as follows.  
Take any $k$ such that $x^* \in X_k$, and let $\sigma$ be a segmentation satisfying $\operatorname{supp} \sigma \subset \{x^S\}_{S \in \mathcal{V}_k}$.
The second is the following direct segmentation $\sigma'$:
\[
\sigma'(x) = 
\begin{cases}
x_k^* & \text{if } x = x^{\{v_k\}}, \\
0 & \text{otherwise}.
\end{cases}
\]
Note that for all $k$, $\sigma'(x^{\{v_k\}})>0$ since we assume that $x^*_k>0$.
Moreover, for a given market $x$, we define the minimum pricing rule $\underline{\phi}(x)$ to deterministically charge $\min \operatorname{supp} x$, and similarly, the maximum pricing rule $\overline{\phi}(x)$ to deterministically charge $\max \operatorname{supp} x$.
Now, both $\underline{\phi}$ and $\overline{\phi}$ are optimal pricing rule for both $\sigma$ and $\sigma'$.\par

Note that any pair $(u, \pi)$ satisfying $u \geq 0$, $\pi > \pi^*$, and $u + \pi \leq w^*$ can be expressed using $\alpha \in (0,1]$ and $\beta \in [0,1]$ as follows:
\[
(u,\pi) = \alpha \cdot (0, w^*) + (1 - \alpha) \cdot \left[ \beta \cdot (w^* - \pi^*, \pi^*) + (1 - \beta) \cdot (0, \pi^*) \right].
\]
According to \cite{bergemann2015limits}, the desired welfare outcome can be achieved using the pair $(\sigma'', \phi)$ defined by:
\[
\sigma''(x) = \alpha \sigma'(x) + (1 - \alpha) \sigma(x),
\]
\[
\phi_k(x) = \beta \, \underline{\phi}_k(x) + (1 - \beta) \, \overline{\phi}_k(x).
\]
Note that $\phi$ is optimal pricing rule for $\sigma''$.\par

Therefore, based on the observation in Section \ref{sec: 観察}, it suffices to show that the constructed $\{x^k\}$, defined as
\[
x^k \triangleq \frac{1}{\sum_{x \in \operatorname{supp} \sigma''} \sigma''(x) \phi_k(x)} \sum_{x \in \operatorname{supp} \sigma''} \sigma''(x) \phi_k(x) \cdot x,
\]
consists of mutually distinct elements.  
In particular, it is sufficient to show that the only optimal price at $x^k$ is $v_k$, which we proceed to prove below.
Since $\phi$ is an optimal pricing rule for $\sigma''$, for any $x$ such that $\sigma''(x) > 0$, if $\phi_k(x) > 0$, then
\[
v_k \sum_{j = k}^K x_j \geq v_i \sum_{j = i}^K x_j \quad \text{for all } i = 1, \ldots, K.
\]
In particular, since $\alpha > 0$, we have $\sigma''(x^{\{v_k\}}) > 0$ and
\[
v_k \sum_{j = k}^K x^{\{v_k\}}_j > v_i \sum_{j = i}^K x^{\{v_k\}}_j \quad \text{for all } i \neq k.
\]
Therefore, for all $i \neq k$, we obtain
\[
\sum_{x \in \operatorname{supp} \sigma''} \sigma''(x) \phi_k(x) \left(v_k \sum_{j = k}^K x_j\right) > \sum_{x \in \operatorname{supp} \sigma''} \sigma''(x) \phi_k(x) \left(v_i \sum_{j = i}^K x_j\right),
\]
or equivalently,
\[
v_k \sum_{j = k}^K x^k_j > v_i \sum_{j = i}^K x^k_j.
\]
\end{proof}

\begin{lem}\label{lem: K points when x^*=x^V}
Let $x^* = x^V$.  
If a direct segmentation achieves producer surplus $\pi = \pi^*$, then the corresponding consumer surplus $u$ satisfies
\[
u = \sum_{j = k}^K (v_j - v_k) x^*_j \quad \text{for some } k.
\]
\end{lem}
\begin{proof}
Take any direct segmentation $\sigma$ (and associated direct pricing rule $\phi$) with producer surplus $\pi = \pi^*$.  
If $|\operatorname{supp} \sigma| = 1$, then we must have $\sigma(x^V) = 1$, and the corresponding consumer surplus $u$ satisfies
\[
u = \sum_{j = k}^K (v_j - v_k) x^*_j \quad \text{for some } k.
\]
Thus, we may assume $|\operatorname{supp} \sigma| \geq 2$.  
Since price $v_1$ is optimal for $x^*$, we have $\pi^* = v_1$.  
Hence, $\pi = \pi^*$ implies
\[
\sum_{x \in \operatorname{supp} \sigma} \sigma(x) 
\sum_{k = 1}^K \phi_k(x) v_k 
\sum_{j = k}^K x_j = v_1.
\]
Now, by the optimality of $\phi$, we have for all $x \in \operatorname{supp} \sigma$,
\[
\sum_{k = 1}^K \phi_k(x) v_k \sum_{j = k}^K x_j \geq v_1.
\]
Therefore, for all $x \in \operatorname{supp} \sigma$,
\[
\sum_{k = 1}^K \phi_k(x) v_k \sum_{j = k}^K x_j = v_1.
\]
It follows that for any $k$ and any $x \in \operatorname{supp} \sigma$,
\[
v_1 \geq v_k \sum_{j = k}^K x_j.
\]
Now, for all $k$,
\[
v_1 = v_k \sum_{j = k}^K x^*_j = \sum_{x \in \operatorname{supp} \sigma} \sigma(x) \left( v_k \sum_{j = k}^K x_j \right).
\]
Therefore, for any $k$ and any $x \in \operatorname{supp} \sigma$,
\[
v_1 = v_k \sum_{j = k}^K x_j.
\]
This implies that every $x \in \operatorname{supp} \sigma$ must satisfy $x = x^V$, which contradicts the assumption that $|\operatorname{supp} \sigma| \geq 2$.

\end{proof}

From Theorem \ref{thm: Surplus Triangle in bergemann2015limits} and Lemmas \ref{lem: pi > pi star} and \ref{lem: K points when x^*=x^V}, we can derive Proposition \ref{prop: direct achievable when $x^* = x^V$}.

\subsection{Proof of Theorem \ref{thm: direct achievable condition}}
The "only if" part follows from Lemma \ref{lem: K points when x^*=x^V}.
Therefore, it suffices to prove the "if" part.  
To this end, we prepare several lemmas.

\begin{lem}\label{lem: x lies conv}
Let $x \in X$, and let $P = \{v_k \in V : \text{price } v_k \text{ is optimal for } x \}$.  
Then,
\[
x \in \operatorname{conv} \left( \{x^S\}_{ S : P \subset S \subset V } \right).
\]
\end{lem}
\begin{proof}
Take any $v_k \in P$.  
By Lemma \ref{lem: convex hull}, we can write
\[
x = \alpha_1 x^{S_1} + \alpha_2 x^{S_2} + \cdots + \alpha_m x^{S_m}
\]
for some $\alpha_i > 0$ and $S_i \in \mathcal{V}_k$, with $\alpha_1+\alpha_2+\cdots+\alpha_m=1$.  
We now prove that $P \subset S_i$ for all $S_i$.  
Suppose, to the contrary, that there exists some $S_{i^*}$ such that $P \setminus S_{i^*} \neq \emptyset$.  
Take any $v_\ell \in P \setminus S_{i^*}$.  
Now, for any $S_i$,
\[
v_k \sum_{j = k}^K x^{S_i}_j \geq v_\ell \sum_{j = l}^K x^{S_i}_j.
\]
Moreover,
\[
v_k \sum_{j = k}^K x^{S_{i^*}}_j > v_\ell \sum_{j = l}^K x^{S_{i^*}}_j.
\]
Thus,
\[
v_k \sum_{j = k}^K x_j=\alpha_1 v_k \sum_{j = k}^K x^{S_1}_j +  \cdots + \alpha_m v_k \sum_{j = k}^K x^{S_m}_j  >\alpha_1 v_\ell \sum_{j = l}^K x^{S_1}_j +  \cdots + \alpha_m v_\ell \sum_{j = l}^K x^{S_m}_j   =v_\ell \sum_{j = l}^K x_j,
\]
which contradicts the assumption that $v_\ell \in P$.  
Therefore, we must have $P \subset S_i$ for all $S_i$, and hence
\[
x \in \operatorname{conv} \left( \{x^S\}_{S : P \subset S \subset V} \right).
\]
\end{proof}

\begin{lem}\label{lem: affine price equal}
Let $P \in \mathcal{V}\setminus\{\emptyset\}$.  
$x \in \operatorname{aff}\left(\{x^S\}_{S : P \subset S \subset V} \right)$
implies
\[
v_k \sum_{j = k}^K x_j = v_\ell \sum_{j = l}^K x_j \quad \text{for all } v_k, v_\ell \in P.
\]
\end{lem}
\begin{proof}
Take any $x \in \operatorname{aff}\left(\{x^S\}_{S : P \subset S \subset V} \right)$ and any $v_k, v_\ell \in P$.  
By the definition of the affine hull, we can write
\[
x = \alpha_1 x^{S_1} + \alpha_2 x^{S_2} + \cdots + \alpha_m x^{S_m}
\]
for some $\alpha_i \in \mathbb{R}$ and $S_i \supset P$.
Since by definition, for any $S_i$,
\[
v_k \sum_{j = k}^K x^{S_i}_j = v_\ell \sum_{j = l}^K x^{S_i}_j,
\]
we have
\[
v_k \sum_{j = k}^K x_j = \alpha_1 v_k \sum_{j = k}^K x^{S_1}_j + \cdots + \alpha_m v_k \sum_{j = k}^K x^{S_m}_j = \alpha_1 v_\ell \sum_{j = l}^K x^{S_1}_j + \cdots + \alpha_m v_\ell \sum_{j = l}^K x^{S_m}_j = v_\ell \sum_{j = l}^K x_j.
\]
\end{proof}

\begin{lem}\label{lem: relative interior}
Let $P = \{v_k \in V : \text{price } v_k \text{ is optimal for } x^* \}$.  
Then, $x^*$ lies in the relative interior of $\operatorname{conv} \left( \{x^S\}_{ S : P \subset S \subset V } \right)$.
\end{lem}
\begin{proof}
Note that $x^*\in \operatorname{conv} \left( \{x^S\}_{ S : P \subset S \subset V } \right)$ by Lemma \ref{lem: x lies conv}, and that
\[
\operatorname{aff} \left( \operatorname{conv} \left( \{x^S\}_{S : P \subset S \subset V} \right) \right) = \operatorname{aff} \left( \{x^S\}_{S : P \subset S \subset V} \right).
\]
Also, since $x^*_k > 0$ for all $k$, there exists some $\varepsilon > 0$ such that for any positive $\delta < \varepsilon$, any $x \in B_\delta(x^*) \cap \operatorname{aff} \left( \{x^S\}_{S : P \subset S \subset V} \right)$ satisfies $x \in \Delta(V)$.
Here, $B_\delta(x^*)$ denotes the open ball $\{ y \in \mathbb{R}^K : \| y - x^* \| < \delta\}$.
We prove that there exists some positive $\delta < \varepsilon$ such that for any $x \in B_{\delta}(x^*) \cap \operatorname{aff} \left( \{x^S\}_{S : P \subset S \subset V} \right)$,
\[
x \in \operatorname{conv} \left( \{x^S\}_{S : P \subset S \subset V} \right).
\]
To derive a contradiction, assume that for every $n\in \mathbb{N}$ satisfying $1/n < \varepsilon$, there exist $x^n \in B_{1/n}(x^*) \cap \operatorname{aff} \left( \{x^S\}_{S : P \subset S \subset V} \right)$ such that $x^n \notin \operatorname{conv} \left( \{x^S\}_{S : P \subset S \subset V} \right)$.
By Lemma \ref{lem: x lies conv} and Lemma \ref{lem: affine price equal}, for each $x^n$,
there exists $v_{k_n} \in V \setminus P$ such that price $v_{k_n}$ is optimal for $x^n$.\par

Since $V \setminus P$ is a finite set, there exists a subsequence $\{v_{k_{a_n}}\}_{n = 1}^{\infty}$ of the sequence $\{v_{k_n}\}_{n = \min \{n'\in \mathbb{N} : 1/\varepsilon < n'\}}^{\infty}$ such that $v_{k_{a_n}} = v_k$ for some fixed $v_k \in V \setminus P$.  
This means that for every $n \in \mathbb{N}$, price $v_k$ is optimal for $x^{a_n}$, that is,
\[
v_k \sum_{j = k}^K x^{a_n}_j \geq v_\ell \sum_{j = l}^K x^{a_n}_j \quad \text{for all } v_\ell \in V, \quad n \in \mathbb{N}.
\]
By construction, we have $\lim_{n \to \infty} x^{a_n} = x^*$, which implies
\[
v_k \sum_{j = k}^K x^{*}_j \geq v_\ell \sum_{j = l}^K x^{*}_j \quad \text{for all } v_\ell \in V.
\]
Thus, price $v_k \notin P$ is optimal for $x^*$, yielding a contradiction.
\end{proof}

\begin{lem}\label{lem: opt price if and if}
Suppose that $x \in X$ can be expressed as
\[
x = \alpha_1 x^{S_1} + \alpha_2 x^{S_2} + \cdots + \alpha_m x^{S_m}
\]
for some $\alpha_i > 0$ and $S_i \in \mathcal{V}_k$.  
Then, price $v_\ell$ is optimal for $x$ if and only if $v_\ell \in S_i$ for all $S_i$.
\end{lem}
\begin{proof}
\textbf{(if part)}  
Suppose that $v_\ell \in S_i$ for all $S_i$.  
Then, for any $k'$, we have
\[
v_{k'} \sum_{j = k'}^K x_j = \alpha_1 v_{k'} \sum_{j = k'}^K x^{S_1}_j + \cdots + \alpha_m v_{k'} \sum_{j = k'}^K x^{S_m}_j \leq \alpha_1 v_\ell \sum_{j = l}^K x^{S_1}_j + \cdots + \alpha_m v_\ell \sum_{j = l}^K x^{S_m}_j = v_\ell \sum_{j = l}^K x_j.
\]
Therefore, price $v_\ell$ is optimal for $x$.\par

\textbf{(only if part)}  
Suppose that price $v_\ell$ is optimal for $x$.  
Assume, for contradiction, that there exists some $S_{i^*}$ such that $v_\ell \notin S_{i^*}$.  
Then, for all $S_i$,
\[
v_{k} \sum_{j = k}^K x^{S_i}_j \geq v_\ell \sum_{j = l}^K x^{S_i}_j,
\]
and moreover,
\[
v_{k} \sum_{j = k}^K x^{S_{i^*}}_j > v_\ell \sum_{j = l}^K x^{S_{i^*}}_j.
\]
Thus,
\[
v_k \sum_{j = k}^K x_j = \alpha_1 v_k \sum_{j = k}^K x^{S_1}_j + \cdots + \alpha_m v_k \sum_{j = k}^K x^{S_m}_j > \alpha_1 v_\ell \sum_{j = l}^K x^{S_1}_j + \cdots + \alpha_m v_\ell \sum_{j = l}^K x^{S_m}_j = v_\ell \sum_{j = l}^K x_j,
\]
which contradicts the optimality of price $v_\ell$ for $x$.
\end{proof}

\begin{lem}\label{lem: change u, no change パイ}
Let $x^*\neq x^V$.
There exists a direct segmentation with consumer surplus $u$ 
and producer surplus $\pi^*$ if and only if $0\leq u \leq w^*-\pi^*$.
\end{lem}
\begin{proof}
The "only if" part follows from Theorem \ref{thm: Surplus Triangle in bergemann2015limits}.  
Therefore, it suffices to prove the "if" part.
Let $P = \{v_k \in V : \text{price } v_k \text{ is optimal for } x^* \}$.  
By Lemma \ref{lem: relative interior}, $x^*$ lies in the relative interior of $\operatorname{conv} \left( \{x^S\}_{ S : P \subset S \subset V } \right)$.  
Therefore, we can write\footnote{See, for example, \cite{rockafellar1997convex}.}
\[
x^* = \sum_{ S : P \subset S \subset V } \alpha_S x^{S}
\]
for some $\alpha_S > 0$, with $\sum_{ S : P \subset S \subset V } \alpha_S = 1$.  
We define the segmentation $\sigma$ by $\sigma(x^S) = \alpha_S$.  
We also consider the minimum pricing rule $\underline{\phi}(x)$ and the maximum pricing rule $\overline{\phi}(x)$, as in the proof of Lemma \ref{lem: pi > pi star} and in \cite{bergemann2015limits}.  
Then, we define the pricing rule $\phi$ by
\[
\phi_k(x) = \frac{u}{w^* - \pi^*} \, \underline{\phi}_k(x) + \frac{w^* - \pi^* - u}{w^* - \pi^*} \, \overline{\phi}_k(x).
\]
By \cite{bergemann2015limits}, $(\sigma, \phi)$ yields consumer surplus $u$ and producer surplus $\pi^*$.  
Note that $\phi$ is an optimal pricing rule for $\sigma$.\par

Based on the observation in Section \ref{sec: 観察}, it suffices to show that the constructed set $\{x^k\}$, defined as
\[
x^k \triangleq \frac{1}{\sum_{x \in \operatorname{supp} \sigma} \sigma(x) \phi_k(x)} \sum_{x \in \operatorname{supp} \sigma} \sigma(x) \phi_k(x) \cdot x=\frac{1}{\sum_{ S : P \subset S \subset V } \alpha_S  \phi_k(x^S)} \sum_{ S : P \subset S \subset V } \alpha_S \phi_k(x^S) \cdot x^S,
\]
consists of mutually distinct elements. 
In what follows, we consider cases separately.\par

We first consider the case where $u = 0$.  
If $x^k$ is well-defined (i.e., $\sum_{ S : P \subset S \subset V } \alpha_S \phi_k(x^S) > 0$), then we have $x^k_k > 0$ and $x^k_j = 0$ for all $j > k$.  
Therefore, the elements of $\{x^k\}$ are all mutually distinct.  
Similarly, in the case where $u = w^* - \pi^*$, if $x^k$ is well-defined, we have $x^k_k > 0$ and $x^k_j = 0$ for all $j < k$.  
Hence, $\{x^k\}$ consists of mutually distinct elements in this case as well.  
Thus, in the following, we focus on the case where $0 < u < w^* - \pi^*$.\par

Next, we consider the case where $P$ is a singleton.
Let $v_{k^*} \in P$.  
For any $v_\ell \neq v_{k^*}$, by Lemma \ref{lem: opt price if and if}, the optimal prices for $x^\ell$ are exactly $v_\ell$ and $v_{k^*}$.  
On the other hand, by Lemma \ref{lem: opt price if and if}, the optimal price for $x^{k^*}$ is uniquely $v_{k^*}$.  
Therefore, the elements of $\{x^k\}$ are all mutually distinct.\par

Next, we consider the case where $P$ is not a singleton, that is, $|P| \geq 2$.  
For any $v_k$ such that $v_k < \min P$ or $v_k > \max P$, by Lemma \ref{lem: opt price if and if}, the set of optimal prices for $x^k$ consists exactly of $v_k$ and the elements of $P$.  
Moreover, for any $v_k$ such that $\min P < v_k < \max P$, $x^k$ is not well-defined.
Finally, by Lemma \ref{lem: opt price if and if}, for $x^{\min P}$ and $x^{\max P}$, both have the same set of optimal prices, namely $P$.  
Thus, it suffices to show that $x^{\min P} \neq x^{\max P}$.\par

First, if $v_1 < \min P$, we have $x^{\min P}_1 = 0$ but $x^{\max P}_1 > 0$, so they are distinct.  
Similarly, if $v_K > \max P$, we have $x^{\max P}_K = 0$ but $x^{\min P}_K > 0$, so they are distinct as well.  
Hence, in what follows, we may assume $v_1 = \min P$ and $v_K = \max P$.\par

Since $x^* \neq x^V$, we have $P \neq V$. Hence, we can take any distinct sets $S_1$ and $S_2$ such that $P \subset S_i \subset V$ for $i = 1, 2$.
In the following, we construct a completely different direct segmentation from those considered so far.  
As a preliminary step, we construct $x'$ and $x''$ as follows:
\[
\begin{aligned}
x' =\; &(\alpha_{S_1} + \alpha_{S_2})
\left\{
\left( \frac{\alpha_{S_1}}{\alpha_{S_1} + \alpha_{S_2}} + (1 - \beta)\varepsilon \right) x^{S_1}
+ \left( \frac{\alpha_{S_2}}{\alpha_{S_1} + \alpha_{S_2}} - (1 - \beta)\varepsilon \right) x^{S_2}
\right\} \\
&+ \sum_{ S \neq S_1, S_2 : P \subset S \subset V } \alpha_S x^{S},
\end{aligned}
\]
\[
x'' = (\alpha_{S_1} + \alpha_{S_2}) \left\{ \left( \frac{\alpha_{S_1}}{\alpha_{S_1} + \alpha_{S_2}} - \beta \varepsilon \right) x^{S_1} + \left( \frac{\alpha_{S_2}}{\alpha_{S_1} + \alpha_{S_2}} + \beta \varepsilon \right) x^{S_2} \right\} + \sum_{ S \neq S_1, S_2 : P \subset S \subset V  } \alpha_S x^{S},
\]
for some $\varepsilon > 0$ and $\beta \in (0,1)$.  
Since $\alpha_{S_1} > 0$ and $\alpha_{S_2} > 0$, there exists $\varepsilon > 0$ such that for any $\beta \in (0,1)$, we have $x', x'' \in X$.  
We take such $\varepsilon > 0$.
Note that for any $\beta\in (0,1)$, $\beta x' + (1 - \beta) x'' = x^*$ and $x' \neq x''$.\par

We construct a direct segmentation $\sigma'$ and an associated direct pricing rule $\phi'$ as follows:
\[
\sigma'(x') = \beta, \quad \sigma'(x'') = 1 - \beta, \quad \phi_1'(x') = 1, \quad \phi_K'(x'') = 1.
\]
Note that by Lemma \ref{lem: opt price if and if}, we have $x' \in X_1$ and $x'' \in X_K$.  
Moreover, since $x'' \in X_1$ as well, the producer surplus under $(\sigma', \phi')$ equals $\pi^*$, regardless of the value of $\beta$.
On the other hand, the consumer surplus
\begin{align*}
u'&=\beta\sum_{j=1}^K (v_j - v_1)x_j'\\
&=\beta\left\{\sum_{j=1}^K (v_j - v_1)x_j^*+(\alpha_{S_1} + \alpha_{S_2}) (1 - \beta) \varepsilon \left( \sum_{j=1}^K (v_j - v_1)x_j^{S_1}-\sum_{j=1}^K (v_j - v_1)x_j^{S_2}\right)\right\}
\end{align*}
depends continuously on $\beta \in (0,1)$, and we have
\[
\lim_{\beta \to 0^+} u' = 0, \quad \lim_{\beta \to 1^-} u' = w^* - \pi^*.
\]
Thus, by the intermediate value theorem, there exists some $\beta \in (0,1)$ such that $u' = u$.
\end{proof}

By Theorem \ref{thm: Surplus Triangle in bergemann2015limits}, Lemma \ref{lem: pi > pi star}, and Lemma \ref{lem: change u, no change パイ}, the "if" part of Theorem \ref{thm: direct achievable condition} is also established.

\printbibliography

\end{document}